\documentclass[12pt]{article}
\usepackage[english]{babel}
\usepackage[utf8]{inputenc}
\usepackage[T1]{fontenc}
\usepackage{listings}
\usepackage{amsbsy}
\usepackage{amsfonts}
\usepackage{amsmath}
\usepackage{amssymb}
\usepackage{amsthm}
\usepackage{graphicx} 
\usepackage[pdftex]{color}
\usepackage{dsfont}
\usepackage{enumerate}
\usepackage{ulem}
\usepackage{float}
\usepackage{geometry, calc, color, setspace}
\usepackage{indentfirst}
\usepackage{longtable}
\usepackage{multirow}
\usepackage{natbib}
\usepackage{xr-hyper}
\usepackage{hyperref,bookmark}

\usepackage{eucal}
\usepackage{algorithm}
\usepackage[algo2e]{algorithm2e}
\usepackage{caption}
\usepackage{subcaption}
\usepackage[toc,page]{appendix}
\usepackage[leftcaption]{sidecap} % for left side caption table

\usepackage{listings}
\usepackage{fancyvrb}
\usepackage{tcolorbox}
\usepackage{placeins} % Keeping floats in their place
\usepackage{xcolor,colortbl}

\newtheorem{theorem}{Theorem}[section]

\newtheorem{proposition}[theorem]{Proposition}

\hypersetup{
  colorlinks=true,
  linkcolor=blue,
  citecolor=blue,
  filecolor=blue,
  urlcolor=blue,
}

\newcommand{\blind}{1}

\protect

\makeatletter
\newcommand*{\bigcdot}{}% Check if undefined
\DeclareRobustCommand*{\bigcdot}{%
  \mathbin{\mathpalette\bigcdot@{}}%
}
\newcommand*{\bigcdot@scalefactor}{.5}
\newcommand*{\bigcdot@widthfactor}{1.15}
\newcommand*{\bigcdot@}[2]{%
  \sbox0{$#1\vcenter{}$}% math axis
  \sbox2{$#1\cdot\m@th$}%
  \hbox to \bigcdot@widthfactor\wd2{%
    \hfil
    \raise\ht0\hbox{%
      \scalebox{\bigcdot@scalefactor}{%
        \lower\ht0\hbox{$#1\bullet\m@th$}%
      }%
    }%
    \hfil
  }%
}
\makeatother

\begin{document}
\def\spacingset#1{\renewcommand{\baselinestretch}%
{#1}\small\normalsize} \spacingset{1}

%\spacingset{1.45} 

\if1\blind
{
  \title{\bf Poisson-Birnbaum-Saunders Regression Model for Clustered Count Data}
  \author{Jussiane Nader Gonçalves$^\star$\footnote{Email: \href{mailto:jussianegoncalves@gmail.com}{jussianegoncalves@gmail.com}},\, Wagner Barreto-Souza$^{\sharp\star}$\footnote{Email: \href{mailto:wagner.barretosouza@kaust.edu.sa}{wagner.barretosouza@kaust.edu.sa}}\, and\, Hernando Ombao$^\sharp$\footnote{Email: \href{mailto:hernando.ombao@kaust.edu.sa}{hernando.ombao@kaust.edu.sa}}\\
    {\it \normalsize $^\star$Departamento de Estat\'istica, Universidade Federal de Minas Gerais, Belo Horizonte, Brazil}\\ 
    {\it \normalsize $^\sharp$Statistics Program, King Abdullah University of Science and Technology, Thuwal, Saudi Arabia}} 
  \date{}
  \maketitle
} \fi
\if0\blind
{
  \bigskip
  \bigskip
  \bigskip
  \begin{center}
    {\LARGE\bf Title}
\end{center}
  \medskip
} \fi

\bigskip
\addtocontents{toc}{\protect\setcounter{tocdepth}{1}}

\begin{abstract}
The premise of independence among subjects in the same cluster/group often fails in practice, and models that rely on such untenable assumption can produce misleading results. To overcome this severe deficiency, we introduce a new regression model to handle overdispersed and correlated clustered counts. To account for correlation within clusters, we propose a Poisson regression model where the observations within the same cluster are driven by the same latent random effect that follows the Birnbaum-Saunders distribution with a parameter that controls the strength of dependence among the individuals. This novel multivariate count model is called Clustered Poisson Birnbaum-Saunders (CPBS) regression. As illustrated in this paper, the CPBS model is analytically tractable, and its moment structure can be explicitly obtained. Estimation of parameters is performed through the maximum likelihood method, and an Expectation-Maximization (EM) algorithm is also developed. Simulation results to evaluate the finite-sample performance of our proposed estimators are presented. We also discuss diagnostic tools for checking model adequacy. An empirical application concerning the number of inpatient admissions by individuals to hospital emergency rooms, from the Medical Expenditure Panel Survey (MEPS) conducted by the United States Agency for Health Research and Quality, illustrates the usefulness of our proposed methodology.\\

\end{abstract}

\section{Introduction}\label{intro}
    Clustered count data is being collected in many sectors and disciplines. In particular, in the actuarial community, it is essential to provide a reliable estimate of the number of claim events in an insurance portfolio to establish a fair rate premium. The current methods for analyzing such data assume that these events are independent. This assumption is certainly unrealistic and likely to produce misleading results. For instance, in auto insurance, the theft claims rate may vary by geographical region, as each area may have its pattern. Some have higher claims rates, and others have lower ones depending on many factors, such as the security level in the province. Another example related to health care and linked to private and social health insurance is the number of admissions in regional hospitals.
    
     Figure \ref{fig1} exhibits the frequency distribution of the number of inpatient admissions by individuals to hospital emergency rooms in US regions from a random sample of the 2003 Medical Expenditure Panel Survey (MEPS). This data set contains health status, access, use, and costs of health services in the USA. This plot shows a noticeable pattern of inpatient admissions which varies across regions in the US. One of our goals is to properly model the number of inpatient admissions according to the geographical US regions as a tool for measuring the volume of diagnostic procedures in the health care system, which could be used to predict future costs related to the needs of the benefited population. A preliminary analysis of the MEPS data set reveals the inadequacies of the current methods, which ignore the correlation within regional clusters. This is the motivation for developing a new model that accounts for the within-cluster correlation among units. 
     One advantage of the proposed model is that it accurately predicts the number of inpatient admissions. This is important since this gives an actuary accurate information for calculating the costs of this significant health insurance component. Moreover, the proposed tool provides the government with information that will be useful in formulating public policy concerning the volume of resources allocated to a public hospital to deal with inpatient admissions.
    
   \begin{figure}[h!]
        \centering
        \includegraphics[width=0.75\textwidth]{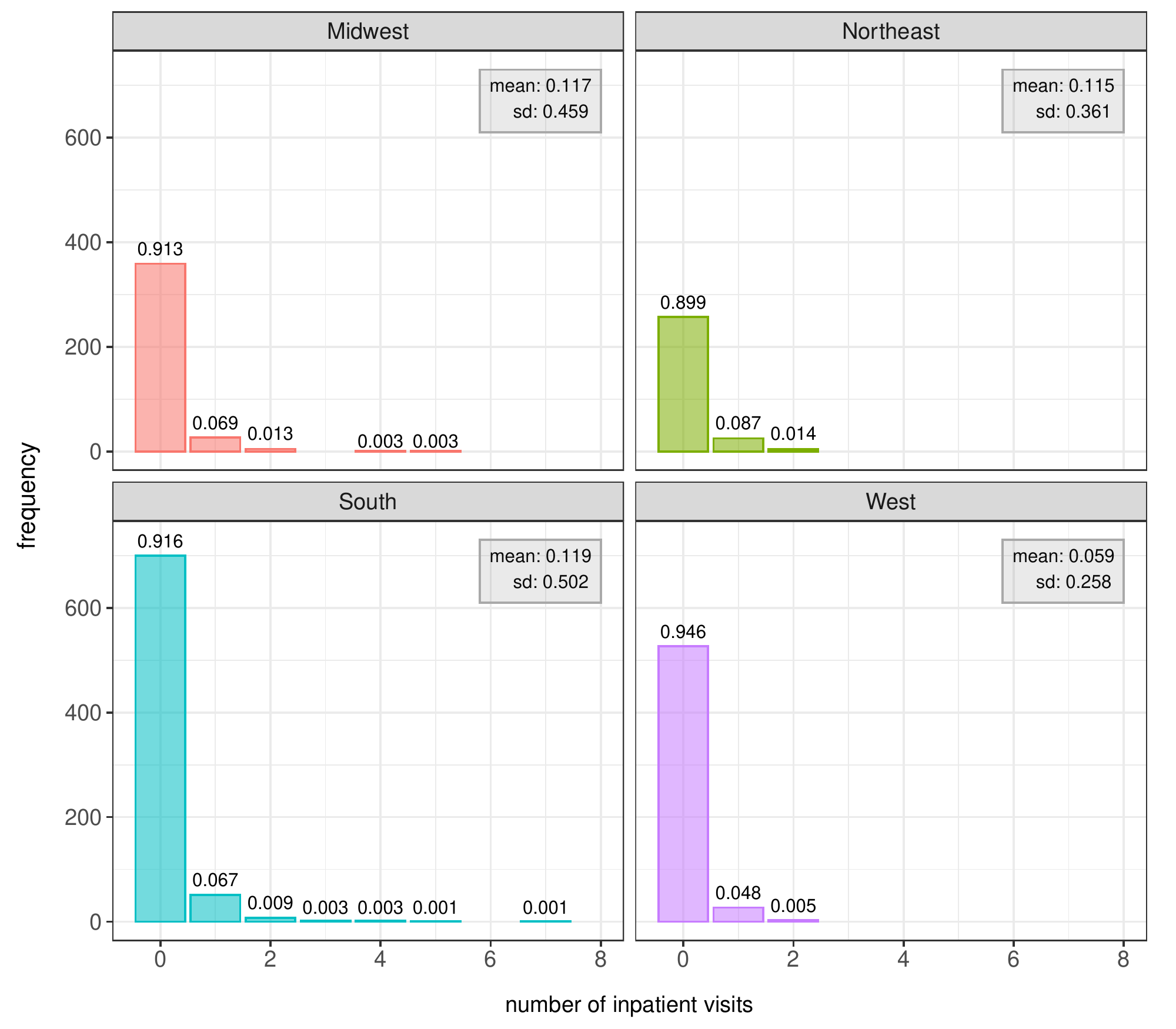}
        \caption{Descriptive data analysis on the number of inpatient admissions to hospital emergency rooms by individuals in US regions, from a random sample of the MEPS study, 2003.}
        \label{fig1}
    \end{figure}
    
    In analyzing count data, the most common approach is to apply the standard Poisson model. However, it is widely known that the Poisson equidispersion (mean equal to variance) premise is usually violated. In fact, to handle the case of overdispersion (variance greater than mean), one may consider the mixed Poisson (MP) models, such as the negative binomial \citep{lawless1987, hilbe2007,cameron2013} and the Poisson-inverse Gaussian \citep{holla1966,willmot1987,dean1989} models; for a unified general class of mixed Poisson regression models with varying dispersion/precision, see \cite{bss2016}. Moreover, to deal with the phenomenon of underdispersion (mean greater than variance) phenomenon, one may use the generalized Poisson \citep{consul1992,famoye2006} and the Conway–Maxwell–Poisson \citep{sellers2010} models. For features and properties of MP models, we refer to the works by \cite{hinde1998} and \cite{sellers2005}. 
    
    When count data has excess or deficit of zeros and the data exhibits the phenomenon of overdispersion or underdispersion, one may use the zero-inflated/deflated models such as the zero-inflated Poisson (ZIP) \citep{lambert1992}, the zero-inflated generalized Poisson (ZIGP) \citep{famoye2006}, and the zero-inflated negative binomial (ZINB) \citep{ridout1998,ridout2001,yau2003} models, among others. A flexible class of regression models for counts with high-inflation of zeros, which contains the ZINB, the zero-inflated Poisson-inverse Gaussian (ZIPIG), and the zero-inflated generalized hyperbolic secant (ZIGHS) models, was proposed by \cite{goncalves2020}. Although these models have played an essential role in modeling count data, they assume that the counts are independent, which might be unrealistic, especially when analyzing clustered or grouped data. Motivated by the need to overcome this limitation, this paper aims to develop models that account for correlation. 
    
    A pragmatic approach to model clustered count data is to include a cluster-specific random intercept in the regression model. Following this  direction, \cite{guo1996} proposed the negative multinomial regression with a random intercept following a gamma distribution and applied it to model the number of transurethral resections of the prostate (performed for Medicare and privately insured patients) in US hospitals. A clustered count regression model with random intercept inverse-Gaussian distributed was proposed by \cite{shoukri2004}. \cite{demidenko2007} compared five inference methods for a Poisson regression for clustered count data, including the standard Poisson regression, the Poisson regression with fixed cluster-specific random effect, a generalized estimating equations technique, an exact generalized estimating equations, and maximum likelihood. In summary, four of the five methods presented similar estimates of the slope coefficients for balanced data, though they showed distinct efficiency in the case of unbalanced data. The author applied the described methods to the number of visits to a doctor after a surgical operation to measure the intensity of medical care, which has considerable variation among hospital regions.

    Other distribution assumptions have been considered for modeling clustered count data.
    In \cite{hall2000}, the ZIP model and the zero-inflated binomial (ZIB) regression models are extended by incorporating cluster-specific random effects. In \cite{yau2003}, Gaussian distributed random effects were used in the linear predictors of the zero-inflated negative binomial mixed model for the length of hospital inpatient stay estimation. A class of zero-inflated clustered count models was proposed in \cite{hall2004} and developed an Expectation-Solution algorithm \citep{rosen2000} to estimate the parameters. Furthermore, a ZIP model with a compound Poisson cluster random effect was proposed by \cite{ma2009}. A Poisson mixed model based on a generalized log-gamma random effect was developed in \cite{fabio2012} which yields a multivariate negative binomial model. The model was used to analyze the number of seizures experienced by epileptic patients and to study the freshwater invertebrate offspring born counts in an aquatic toxicology experiment. Recent contributions on the analysis of clustered count data are due to \cite{choowosoba2016}, \cite{choodatta2018}, \cite{choowosoba2018}, and \cite{kang2021}.
    
    The primary goal in this paper is to develop a novel count multivariate model, which is a Poisson regression model with cluster-specific random effects following a Birnbaum-Saunders distribution \citep{bs1969}. We call this novel model the Clustered Poisson Birnbaum-Saunders (CPBS) regression model. We will demonstrate some of the advantages of the CPBS model: it is analytically tractable, and its moment structure can be explicitly derived. Moreover, we will show the explicit form of the likelihood function from which we obtain the maximum likelihood estimator. This is a remarkable feature over some existing clustered count models where the likelihood function is not obtained explicitly, and then approximations or computationally demanding algorithms are necessary to perform inference. Our idea is that our methodology can be considered as an additional tool to the current methods when analyzing such type of data, especially under the current era in data science and machine learning where multiple models can be considered to deliver the best prediction as possible, mainly when explicit knowledge on the mechanism behind the outcome of interest is not fully known. Other contributions of the present paper are the following: (i) the development of an Expectation-Maximization (EM) algorithm \citep{dempster1977} to estimate parameters when numerical issues are experienced when performing the direct maximization of the log-likelihood function due to its dependency on the Bessel function (more details are provided in Section \ref{EM}); (ii) complete statistical analysis including diagnostic tools for checking model adequacy; (iii) application of the proposed CBPS model to the Medical Expenditure Panel Survey (MEPS) data where the assumption of independence can deliver different conclusions when compared to the cluster-based analysis. 
    
In Section \ref{modelspec}, we introduce the multivariate/clustered Poisson-Birnbaum-Saunders regression model and obtain its moment structure and likelihood function in closed forms. In Section \ref{EM}, we discuss an estimation procedure based on the maximum likelihood method and develop an EM-algorithm to estimate the model parameters. We also develop a procedure for computing the standard errors of the parameter estimates. Section \ref{diagnostic} is dedicated to diagnostic tools, including a residual analysis based on simulated envelopes and the derivation of the Cook's distance to identify possible influential observations. Section \ref{simulation} presents simulated results that confirm a good finite-sample performance of the proposed estimators. A statistical analysis of the number of inpatient admissions by individuals to hospital emergency rooms from the MEPS study based on our CPBS regression is presented in Section \ref{illustration}. Concluding remarks and future research are drawn in Section \ref{conclusion}.
    
\section{Model specification}\label{modelspec}
Denote by $Y_{kj}$ the count of the $j$-th individual from the $k$-th cluster, for $k = 1, \ldots, q$, 
and $j=1,\ldots, n_k$, where $n_k$ is the number of individuals in the $k$-th cluster and $q$ is the number of clusters. The total sample size is $n=\sum_{j=1}^qn_j$. To accommodate correlation among the counts with-in the clusters, we consider a sequence of independent and identically distributed random variables $T_1,\ldots,T_q$ following a Birnbaum-Saunders (BS) distribution with scale parameter to 1 (to avoid non-identifiability problems) and shape parameter $\phi \in  (0, \infty)$, with probability density function 
    \begin{equation}
        f\left(t\right) = \dfrac{t^{-1/2}+t^{-3/2}}{2\sqrt{2\pi}\phi} \exp\left(-\dfrac{t + t^{-1} - 2}{2\phi^2}\right), \quad t>0.
        \label{ms2}
    \end{equation}
Here, denote $T_k\sim\mbox{BS}(\phi)$. The mean and variance are given by $E(T_k)=\left(1+\phi^2/2\right)$ and $\mbox{Var}(T_k)=\phi^2\left(1 + 5\phi^2/4\right)$, respectively. 
The Clustered Poisson-Birnbaum-Saunders (CPBS) regression model is defined by assuming that (i) the counts of individuals belonging to different clusters are independent, that is $Y_{k i}\perp Y_{l j}$ for all $k\neq l$, $i=1,\ldots,n_k$ and $j=1,\ldots,n_l$;\  and (ii) $Y_{k1},\ldots,Y_{kn_k}$ are conditionally independent given $T_k$ and satisfy the stochastic representation
    \begin{eqnarray}\label{storep}
    Y_{kj}|T_k \sim \mbox{Poisson}(\mu_{kj} T_k),
    \end{eqnarray}
for $j = 1, \dots, n_k$ and $k = 1, \ldots, q$, with the $\mu_{kj}$'s being location parameters with the following regression structure:
    \begin{flalign}
        &g(\mu_{kj}) = \boldsymbol{x}^{\top}_{kj}\boldsymbol{\beta},
        \label{ms1}
    \end{flalign}
where $g(\cdot)$ is an invertible link function ensuring the location parameters are positive, $\boldsymbol{x}_{kj} = \left(x_{kj1}, \dots, x_{kjp}\right)^{\top}$ stands for the $p \times 1$ vector of explanatory variables/covariates related to the $j$-th individual from the $k$-th cluster and $\boldsymbol{\beta}=(\beta_1,\ldots,\beta_p)^\top$ is an associated parameter vector. Moreover, the matrix $\bf X$ composed by covariate vectors is assumed to have full rank. The motivation for considering a regression structure here comes from the fact that, in numerous practical situations, covariates are available and informative for studying the (conditional) distributions of the outcomes of interest \cite{cameron2013}. One expects variations or distributions in the pattern of the number of inpatient admissions to change with age. For example, younger subpopulations have a lower utilization rate than the elderly subpopulation.
    
From the assumption given in (\ref{storep}), note that the conditional probability function of the random vector $(Y_{k1}, \dots, Y_{kn_k})$ given $T_k$ is given by
    \begin{equation}
        p\left(y_{k1}, \dots, y_{kn_{k}}|t_{k}\right) = \prod\limits_{j=1}^{n_k} \dfrac{e^{-\mu_{kj}t_{k}}\left(\mu_{kj}t_{k}\right)^{y_{kj}}}{y_{kj}!}, \quad y_{kj} \in \mathbb{N}, \quad \mu_{kj} > 0,
        \label{ms3}
    \end{equation}
for $j = 1, \dots, n_k$, and $k = 1, \dots, q$. In the following proposition, we provide the joint probability function of $Y_{k1},\ldots,Y_{kn_k}$ (counts from the $k$-th cluster), which will enable us to perform maximum likelihood estimation of parameters via direct maximization and also through an EM-algorithm in the next section.
    
\begin{proposition}\label{jpf}
For $k = 1,\ldots, q$, the joint probability function of $Y_{k1},\ldots,Y_{kn_k}$ assumes the form
    \begin{eqnarray}
        p\left(y_{k1}, \dots, y_{kn_k}\right)  
         = \dfrac{e^{1/\phi^2}}{\sqrt{2\pi}\phi}  \left(\prod\limits_{j=1}^{n_k} \dfrac{\mu_{kj}^{y_{kj}}}{y_{kj}!}\right)  \left\{\dfrac{\mathcal{K}_{y_{k\bigcdot}+\frac{1}{2}}\left(\dfrac{\sqrt{1+2\phi^2\mu_{k\bigcdot}}}{\phi^2}\right)}{\left(1+2\phi^2\mu_{k\bigcdot}\right)^{(y_{k\bigcdot}+1/2)/2}} + \dfrac{\mathcal{K}_{y_{k\bigcdot}-\frac{1}{2}}\left(\dfrac{\sqrt{1+2\phi^2\mu_{k\bigcdot}}}{\phi^2}\right)}{\left(1+2\phi^2\mu_{k\bigcdot}\right)^{(y_{k\bigcdot}-1/2)/2}}\right\}, 
        \label{ms4}
    \end{eqnarray}
for $y_{k1}, \ldots, y_{kn_k}\in \mathbb N$, where $y_{k\bigcdot} \equiv \sum\limits_{j=1}^{n_k} y_{kj}$, $\mu_{k\bigcdot} \equiv \sum\limits_{j=1}^{n_k} \mu_{kj}$, and
$$\mathcal{K}_{\lambda}\left(\omega\right) \equiv \dfrac{1}{2} \int\limits_{0}^{\infty} u^{\lambda-1} \exp \left\{-\dfrac{\omega}{2}\left(u + \dfrac{1}{u}\right)\right\}du, \quad \omega > 0, \ \ \lambda \in \mathbb{R},$$
is the modified Bessel function of the third kind.  
\end{proposition}
    
\begin{proof}
We have that
\begin{eqnarray}\label{aux1}
    p\left(y_{k1}, \dots, y_{kn_k}\right)  
         &=&\int\limits_{0}^{\infty} f\left(y_{k1}, \dots, y_{kn_k}, t_{k}\right) dt_{k} 
         =\int\limits_{0}^{\infty} p\left(y_{k1}, \dots, y_{kn_k}| t_{k}\right)f(t_{k}) dt_{k}\nonumber \\
         &=& \left(\prod\limits_{j=1}^{n_k} \dfrac{\mu_{kj}^{y_{kj}}}{y_{kj}!}\right)\int\limits_{0}^{\infty}e^{-\mu_{k\bigcdot} t_k}t_k^{y_{k\bigcdot}}f(t_k)dt_k\nonumber\\
           &=& \dfrac{e^{1/\phi^2}}{2\sqrt{2\pi}\phi}\left(\prod\limits_{j=1}^{n_k} \dfrac{\mu_{kj}^{y_{kj}}}{y_{kj}!}\right)\Bigg\{\int\limits_{0}^{\infty}t_k^{y_{k\bigcdot}-1/2}\exp\left\{-\dfrac{1}{2}\left[t_k(2\mu_{k\bigcdot}+\phi^{-2})+t_k^{-1}\phi^{-2}\right]\right\}dt_k\nonumber \\
          &&+\int\limits_{0}^{\infty}t_k^{y_{k\bigcdot}-3/2}\exp\left\{-\dfrac{1}{2}\left[t_k(2\mu_{k\bigcdot}+\phi^{-2})+t_k^{-1}\phi^{-2}\right]\right\}dt_k\Bigg\},
    \end{eqnarray}
    where we have used (\ref{ms2}) and (\ref{ms3}). The above integrals can be solved by identify density kernels of generalized inverse Gaussian (GIG) distributions. We say that a random variable follows a GIG distribution with parameters $a,b>0$ and $\alpha\in\mathbb R$ if its density function is given by
    \begin{eqnarray*}
    h(z)=\dfrac{(a/b)^{\alpha/2}}{2\mathcal{K}_\alpha(\sqrt{ab})}z^{\alpha-1}\exp\{-(az+b/z)/2\},\quad z>0.
    \end{eqnarray*}
    
    Then,
    \begin{eqnarray}\label{aux2}
    \int_0^\infty z^{\alpha-1}\exp\{-(az+b/z)/2\}dz=2(b/a)^{\alpha/2}\mathcal{K}_\alpha(\sqrt{ab}).
    \end{eqnarray}
The first and the second integrals in (\ref{aux1}) are obtained from (\ref{aux2}) with $(a,b,\alpha)=(2\mu_{k\bigcdot}+\phi^{-2},\phi^{-2},y_{k\bigcdot}+1/2)$ and $(a,b,\alpha)=(2\mu_{k\bigcdot}+\phi^{-2},\phi^{-2},y_{k\bigcdot}-1/2)$, respectively. This gives us the desired result.
\end{proof}
       
The following result provides an explicit form for the moment structure of the proposed CPBS regression model.  
\begin{proposition}
The moment structure of a CPBS model is given by
    \begin{flalign*}
        E(Y_{kj}) &= \mu_{kj}\left(1 + \dfrac{\phi^2}{2}\right),\\[0.5cm]
        \mbox{Var}(Y_{kj}) &= \mu_{kj}\left(1 + \dfrac{\phi^2}{2}\right) + \mu_{kj}^2 \phi^2 \left(1 + \dfrac{5}{4}\phi^2\right), \ \text{and}\\[0.5cm]
        \mbox{Cov}(Y_{ki},Y_{kj}) &= \mu_{ki} \mu_{kj} \phi^2 \left(1 + \dfrac{5}{4}\phi^2\right),
    \end{flalign*}
    for $j = 1, \dots, n_{k}$, and $k = 1, \dots, q$.
    \end{proposition}
    
    \begin{proof}
      By using properties of conditional mean, variance, and covariance, and the two first cumulant of BS distribution, we have that
    \begin{eqnarray*}
        E(Y_{kj}) = E[E(Y_{kj}|T_{k})]= E(\mu_{kj}T_{k})= \mu_{kj}\left(1 + \dfrac{\phi^2}{2}\right),
    \end{eqnarray*}
     \begin{eqnarray*}  
        \mbox{Var}(Y_{kj})&=&E[\mbox{Var}(Y_{kj}|T_{k})] + \mbox{Var}[E(Y_{kj}|T_{k})]
            = E(\mu_{kj}T_{k}) + \mbox{Var}(\mu_{kj}T_{k})\\
           &=& \mu_{kj}\left(1 + \dfrac{\phi^2}{2}\right) + \mu_{kj}^2 \phi^2 \left(1 + \dfrac{5}{4}\phi^2\right),
    \end{eqnarray*}    
    and
         \begin{eqnarray*}    
        \mbox{Cov}(Y_{ki},Y_{kj}) &=& E[\mbox{Cov}(Y_{ki},Y_{kj}|T_{k})] + \mbox{Cov}[E(Y_{ki}|T_{k}),E(Y_{kj}|T_{k})]
            =  0 + \mbox{Cov}(\mu_{ki}T_{k}, \mu_{kj}T_{k})\\
           &=& \mu_{ki}\mu_{kj} \mbox{Var}(T_k)
            = \mu_{ki} \mu_{kj} \phi^2 \left(1 + \dfrac{5}{4}\phi^2\right).
    \end{eqnarray*}

    \end{proof}

    We conclude this section by highlighting that the univariate Poisson-Birnbaum-Saunders (PBS) distribution (case $n_j=1$ for $j=1,\ldots,q$) have already appeared in the literature. The univariate Poisson-mixed inverse Gaussian class of distributions by \cite{gupta2016} contains the PBS distribution as a particular case. On the other hand, novel contributions of our proposed methodology are both dependence modeling and allowance for covariates to explain the variation of the distribution of the counts.

\section{Likelihood inference}\label{EM}
    
    In this section, we discuss the estimation of parameters of the CPBS regression model through the maximum likelihood method. Denote by $\boldsymbol{\theta} = \left(\boldsymbol{\beta}^{\top}, \phi\right)$ the parameter vector. The log-likelihood function is $\ell(\boldsymbol\theta)=\sum_{k=1}^q\log  p\left(y_{k1}, \dots, y_{kn_{k}}\right)$, with $p(\cdot)$ as given in (\ref{ms4}). More explicitly, we have that
    \begin{eqnarray}
        \ell(\boldsymbol\theta) &\propto&  q(\phi^{-2} - \log \phi) + \sum\limits_{k=1}^{q} \sum\limits_{j=1}^{n_k} y_{kj} \log \mu_{kj} +\nonumber\\
        &&\sum\limits_{k=1}^{q}\log\left(\dfrac{\mathcal{K}_{y_{k\bigcdot}+\frac{1}{2}}\left(\dfrac{\sqrt{1+2\phi^2\mu_{k\bigcdot}}}{\phi^2}\right)}{\left(1+2\phi^2\mu_{k\bigcdot}\right)^{(y_{k\bigcdot}+1/2)/2}} 
        +  \dfrac{\mathcal{K}_{y_{k\bigcdot}-\frac{1}{2}}\left(\dfrac{\sqrt{1+2\phi^2\mu_{k\bigcdot}}}{\phi^2}\right)}{\left(1+2\phi^2\mu_{k\bigcdot}\right)^{(y_{k\bigcdot}-1/2)/2}} \right).
        \label{ms5}
    \end{eqnarray}
    
   The Bessel function involved in the likelihood function can be computed in software such as the $\mathtt{R}$, $\mathtt{MAPLE}$, and $\mathtt{MATHEMATICA}$. The maximum likelihood estimator of $\boldsymbol\theta$ is $\widehat{\boldsymbol\theta}=\mbox{argmax}_{\boldsymbol\theta}\ell(\boldsymbol\theta)$, which can be obtained numerically through some optimization algorithm such as BFGS. The standard errors of the maximum likelihood estimates can be obtained directly from the Hessian matrix.
   
   In our numerical experiments, we encountered  numerical issues in the optimization of the log-likelihood function (\ref{ms5}) due to Bessel functions. To overcome this problem, we develop an EM-algorithm \citep{dempster1977}, where the maximization step involves a simpler function to be optimized. 
    
Let $\left\{\left(Y_{k1}, \dots, Y_{kn_k}, T_{k}\right)\right\}_{k=1}^{q}$ be the complete data, where the $Y_{kj}$'s are the observable counts and the $T_{k}$'s are latent (non-observable) Birnbaum-Saunders random effects. The complete log-likelihood function is
    \begin{flalign}
        \ell_{c}(\boldsymbol{\theta})  \propto q(\phi^{-2} - \log \phi) + \sum\limits_{k=1}^{q}\left\{ \sum\limits_{j=1}^{n_k} y_{kj} \log \mu_{kj}  - \left(\mu_{k\bigcdot} + \dfrac{1}{2\phi^2}\right) t_{k} -\dfrac{t^{-1}_{k}}{2\phi^2} \right\}.
        \label{ms6}
    \end{flalign}

 In what follows, we develop the two steps required by the EM-algorithm with details.

\subsection{Expectation step}
    We now develop the E-step of the algorithm which consists of computing the conditional expectation of the complete log-likelihood function given the data also known as $Q$-function: $Q(\boldsymbol\theta;\boldsymbol\theta^{(r)})=E(\ell_{c}(\boldsymbol{\theta})|{\bf Y}; \boldsymbol\theta^{(r)})$, where ${\bf Y}$ denotes all the observable counts, $\boldsymbol\theta^{(r)}$ is the EM-estimate of $\boldsymbol\theta$ in the $r$-th iteration of the algorithm. The next proposition gives us the conditional expectations to compute the $Q$-function.

    \begin{proposition}
        \label{prop1}
        For $k=1,\ldots,q$, the conditional moments of $T_k$ given the counts with-in the $k$-th cluster are given by 
        \begin{multline*}
            E(T_{k}^{s}|Y_{k1}=y_{k1}, \ldots,Y_{kn_k}=y_{kn_k}) = \dfrac{1}{p\left(y_{k1}, \ldots, y_{kn_k}\right)}  \dfrac{e^{1/\phi^2}}{\sqrt{2\pi}\phi}  \left(\prod\limits_{j=1}^{n_k} \dfrac{\mu_{kj}^{y_{kj}}}{y_{kj}!}\right) \\ \times \left\{ \dfrac{\mathcal{K}_{y_{k\bigcdot}+\frac{1}{2}+s}\left(\dfrac{\sqrt{1+2\phi^2\mu_{k\bigcdot}}}{\phi^2}\right)}{\left(1+2\phi^2\mu_{k\bigcdot}\right)^{(y_{k\bigcdot}+1/2+s)/2}} + \dfrac{\mathcal{K}_{y_{k\bigcdot}-\frac{1}{2}+s}\left(\dfrac{\sqrt{1+2\phi^2\mu_{k\bigcdot}}}{\phi^2}\right)}{\left(1+2\phi^2\mu_{k\bigcdot}\right)^{(y_{k\bigcdot}-1/2+s)/2}} \right\},
        \end{multline*}
        for $s \in \mathbb R$,  where  $y_{k\bigcdot} = \sum\limits_{j=1}^{n_k} y_{kj}$,  $\mu_{k\bigcdot} = \sum\limits_{j=1}^{n_k} \mu_{kj}$,  and  $p\left(y_{k1}, \ldots, y_{kn_k}\right)$ is given in (\ref{ms4}).
    \end{proposition}
    \begin{proof}
     We have that $$E(T_{k}^{s}|Y_{k1}=y_{k1}, \ldots,Y_{kn_k}=y_{kn_k})=\int_0^\infty t_k^s\,  p\left(y_{k1}, \dots, y_{kn_k}| t_{k}\right)f(t_{k}) dt_{k}/p(y_{k1}, \ldots,y_{kn_k}),$$ where the integral can be solved by following the same steps of proof of Proposition \ref{jpf} (identification of GIG kernels) and therefore the details are omitted.
    \end{proof}
    
    The $Q$-function is obtained by assessing the conditional expectation of the complete log-likelihood in (\ref{ms6}), which is possible to evaluate applying Proposition \ref{prop1}. Thus, its expression is given by
    \begin{flalign}
        Q(\boldsymbol{\theta}; \boldsymbol{\theta}^{(r)})\propto q(\phi^{-2} - \log \phi) + \sum\limits_{k=1}^{q}\left\{ \sum\limits_{j=1}^{n_k} y_{kj} \log \mu_{kj} - \left(\mu_{k\bigcdot} + \dfrac{1}{2\phi^2}\right) \delta_{k}^{(r)} -\dfrac{\gamma_{k}^{(r)}}{2\phi^2} \right\}, 
        \label{ms7}
    \end{flalign}
    for \ $j = 1, \dots, n_{k}$, and $\boldsymbol{\theta}^{(r)}$ being the estimate of $\boldsymbol{\theta}$ in the $r$th loop of the EM-algorithm, where we have defined $\delta_{k}^{(r)} = E(T_{k}|Y_{k1}=y_{k1}, \ldots,Y_{kn_k}=y_{kn_k}; \boldsymbol{\theta}^{(r)})$ and $\gamma_{k}^{(r)} = E(T_{k}^{-1}|Y_{k1}=y_{k1}, \ldots,Y_{kn_k}=y_{kn_k}; \boldsymbol{\theta}^{(r)})$, for $k=1,\ldots,q$, with explicit expressions obtained from Proposition \ref{prop1} with $s=1$ and $s=-1$, respectively.

\subsection{Maximization step}

    Next, we develop the M-step which aims to maximize the $Q$-function. Using the current estimate of the parameters, say $\boldsymbol{\theta}^{(r)}$, it updates the $Q$-function through the conditional expectations $\delta_{k}^{(r)}$, $\gamma_{k}^{(r)}$, and maximizes it again, getting $\boldsymbol{\theta}^{(r+1)} = \mbox{argmax}_{\boldsymbol{\theta}} \ Q(\boldsymbol{\theta}; \boldsymbol{\theta}^{(r)})$. This process, until a settled convergence criterion is satisfied, will be repeated.
    The $Q$-function was implemented in the $\mathtt{R}$ environment \citep{R2021} to perform the EM-algorithm for model inference since the $Q$-function maximization does not have a closed-form solution. For the optimization procedure, the \texttt{nlm} function in the \texttt{stats} package, from the $\mathtt{R}$ program, is used, operating a Newton-type method. The score function associated with the $Q$-function (\ref{ms7}) is given by
    \begin{flalign}
        \dfrac{\partial Q(\boldsymbol{\theta}; \boldsymbol{\theta}^{(r)})}{\partial \beta_{l}} &= \sum\limits_{k=1}^{q} \sum\limits_{j=1}^{n_k} \left(y_{kj} - \delta_{k}^{(r)} \mu_{kj} \right) x_{kjl}, \qquad \text{for} \ \ l = 1, \dots, p, \ \ \text{and} \label{gradbeta} \\[0.5cm]
        \dfrac{\partial Q(\boldsymbol{\theta}; \boldsymbol{\theta}^{(r)})}{\partial \phi} &= \sum\limits_{k=1}^{q} \left\{ \dfrac{1}{\phi^3} \left(\delta_{k}^{(r)} + \gamma_{k}^{(r)} - 2\right) - \dfrac{1}{\phi} \right\}.\nonumber
        \label{ms8}
    \end{flalign}
    
    Note that the $\beta_l$'s can be estimated independently from $\phi$ in each step of the EM-algorithm. Moreover, from (\ref{gradbeta}), we obtain that their EM estimates in each step can be obtained from a Poisson regression fit with offsets $\log\delta_{k}^{(r)}$, $k=1,\ldots,q$. Equating (\ref{ms8}) to zero, we find that the EM estimate of $\phi$ is given in a closed-form as follows:
    \begin{eqnarray}\label{EM-phi}
    \phi^{(r+1)}=\sqrt{\sum_{k=1}^q(\delta_{k}^{(r)} + \gamma_{k}^{(r)})\big/q-2}.
    \end{eqnarray}
    
    In short, we have that the optimization procedure required to perform the EM-estimation of the CPBS regression relies on a Poisson regression fit in each step to obtain $\beta_l^{(r)}$, for $l=1,\ldots,p$, and the EM-estimate for $\phi$ given analytically by (\ref{EM-phi}). This is much simpler, computationally speaking than maximizing (\ref{ms5}). A description of the EM procedure estimation is provided in Algorithm \ref{Algorithm 1}.
    
    \begin{algorithm}[h!]
        \caption{\vspace{0.1cm} EM-algorithm for the CPBS regression model}
        \label{Algorithm 1}
        \vspace{0.2cm}
        \hspace{0.3cm}
        \begin{minipage}{0.95\textwidth}
            1. Choose some initial value for $\boldsymbol{\theta}$, say $\boldsymbol{\theta}^{(0)}$, to start the algorithm.\; \\[0.3cm]
            2. {\bf E-step}: utilizing $\boldsymbol{\theta}^{(r)}$ (the estimate of $\boldsymbol{\theta}$ in the $r$th step), update the $Q$-function by means of $\delta_{k}^{(r)}$ and $\gamma_{k}^{(r)}$ obtained from Proposition \ref{prop1}, for $k=1,\ldots,q$.\; \\[0.3cm]
            3. {\bf M-step}: find the maximum global point of the $Q$-function, say $\boldsymbol{\theta}^{(r+1)}$,  by equating (\ref{gradbeta}) to zero, and using (\ref{EM-phi}).\; \\[0.3cm]
            4. Check if the settled convergence criterion is satisfied. For example, one could use $\max\{||Q(\boldsymbol{\theta}^{(r+1)}; \boldsymbol{\theta}^{(r)}) - Q(\boldsymbol{\theta}^{(r)}; \boldsymbol{\theta}^{(r)})||, ||\boldsymbol{\theta}^{(r+1)} - \boldsymbol{\theta}^{(r)}||\} < \epsilon$. If it is validated, the estimate of $\boldsymbol{\theta}$ is $\boldsymbol{\widehat{\theta}} = \boldsymbol{\theta}^{(r+1)}$. Otherwise, update $\boldsymbol{\theta}^{(r)}$ by $\boldsymbol{\theta}^{(r+1)}$ and go back to E-step.\;\\[-0.2cm]
        \end{minipage}
    \end{algorithm}
 
    According to \cite{louis1982}, when working with the EM-algorithm, the observed information matrix can be derived by
    \begin{eqnarray}\label{ms8}
        I(\boldsymbol{\theta}) = E\left(-\dfrac{\partial^2 \ell_{c}(\boldsymbol{\theta})}{\partial \boldsymbol{\theta} \partial \boldsymbol{\theta}^\top}\Bigg| \ \boldsymbol{Y}\right) - E\left(\dfrac{\partial \ell_{c}(\boldsymbol{\theta})}{\partial \boldsymbol{\theta}} {\dfrac{\partial \ell_{c}(\boldsymbol{\theta})}{\partial \boldsymbol{\theta}}}^\top \Bigg| \ \boldsymbol{Y}\right),
    \end{eqnarray}
   where $\boldsymbol{Y}$ represents the observed data. The elements of the information matrix (\ref{ms8}) based on the EM-approach will not be operated, in this work, to obtain the standard errors of the model parameter estimates, seeing that it wraps a frame with multidimensional arrays, representing a highly unwieldy computational process.
    
    The bootstrap resampling method, introduced by \cite{efron1979}, is a powerful computational technique to construct a sampling distribution of a statistic emanated from a random sample. Thus, we shall develop a bootstrap-based resampling method for producing standard errors of the estimates of the proposed CPBS model parameters, sidestepping the intricate numerical offshoots of the information matrix (\ref{ms8}) and the ungainly computational procedure. In short, for a parametric bootstrap, we assume that the population comes from a CPBS model and draw $B$ samples of $q$ clusters with sizes $n_k$, for $k = 1, \dots, q$. Then, we compute the maximum likelihood estimates of $\boldsymbol{\theta}$, based on $Q$-function (\ref{ms7}), for each one. The sample standard errors of these $B$ values estimate the standard errors of $\boldsymbol{\widehat{\theta}}$. For more details of bootstrap techniques, see \cite{efron1994}.
    
    A Monte Carlo simulation study will be presented in Section \ref{simulation} to assess the finite-sample behavior of estimators based on the EM-approach. Diagnostic tools concerning the clustered PBS regression will be addressed in the next section.
    
\section{Residual and influence diagnostic}\label{diagnostic}
    
    The cycle of the model specification, to analyze a set of count data, includes estimation, testing, and evaluation. To reach the last step, one might perform residual analysis and use goodness-of-fit measures. According to \cite{cameron2013}, the practitioner carries out the residual analysis for many purposes, such as to detect model misspecification, outliers, poor fit, and influential observations. Consequently, residual analysis is pretty essential, and the techniques to perform it will measure the departure between the fitted and the original values of the dependent variable. Besides, a visual analysis may potentially indicate the nature of misspecification and the magnitude of its effect. As count models do not have a single residual definition, and the literature has proposed miscellaneous residuals for count data, following one of the approaches presented by \cite{cameron2013}, we use here the Pearson residual, also known as standardized residual, which is defined by
    \begin{equation}
        r_{kj} = \dfrac{y_{kj}-\widehat{\lambda}_{kj}}{\sqrt{\widehat{\sigma}^2_{kj}}}, \qquad j = 1, \dots, n_k, \quad k = 1, \dots, q,
        \label{rid1}
    \end{equation}
    where
    $$\widehat{\lambda}_{kj} = g^{-1}\left(\boldsymbol{x}^{\top}_{kj}\boldsymbol{\widehat{\beta}}\right)\left(1 + \dfrac{\widehat{\phi}^2}{2}\right), \ \ \mbox{and}$$
    $$\widehat{\sigma}^2_{kj} = \widehat{\lambda}_{kj} + \left[g^{-1}\left(\boldsymbol{x}^{\top}_{kj}\boldsymbol{\widehat{\beta}}\right) \widehat{\phi}\right]^2 \left(1 + \dfrac{5}{4}\widehat{\phi}^2\right),$$
    with $\boldsymbol{\widehat{\beta}}$ and $\widehat{\phi}$ being the maximum likelihood estimates (MLEs) of $\boldsymbol{\beta}$ and $\phi$, respectively, obtained through the EM-algorithm or via a direct maximization of (\ref{ms5}).
    
    An ordinary way to employ residuals is to plot them against the normal quantiles. However, even though Pearson's residuals have zero mean and unit variance for large samples, they are skewed in distribution. Therefore, we expect a poor normal approximation, even for moderate sample sizes. To overcome this barrier, we will construct simulated envelopes for the residuals, as suggested by \cite{atkinson1985}, and \cite{hinde1998}. The steps to produce simulated envelopes for count regression models follow the description of Algorithm \ref{Algorithm2}. In this way, we will exemplify the effectiveness of these simulated envelopes in the empirical illustration in Section \ref{illustration}.
    
    \begin{algorithm}[h!]
            \caption{\vspace{0.1cm} \bf Simulated envelopes for residuals}
            \label{Algorithm2}
            \vspace{0.2cm}
            \For{$k = 1$ \KwTo $q$}{
            \For{$j = 1$ \KwTo $n_k$}{
                1. Compute $\widehat{\mu}_{kj}$ and $\widehat{\phi}$.\;\\[0.1cm]
                2. Generate $n_k$ observations of $\tilde{Y}_{kj}$, where $\tilde{Y}_{kj} \sim \mbox{CPBS}(\widehat{\mu}_{kj}, \widehat{\phi})$.\;\\[0.1cm]
                3. Obtain the regression coefficients $\boldsymbol{\tilde{\theta}} = \boldsymbol{\tilde{\beta}}$ from the regression of ${\bf \tilde{Y}}_k$ on the covariates.\;\\[0.1cm]
                4. Compute Pearson residuals using $\tilde{Y}_{kj}$ and (\ref{rid1}), denoting the yield residual by $\tilde{R}_{kj}$.\;
            }
            }
            \vskip 0.2cm
            Let $N = \sum\limits_{k=1}^{q} n_k$ and repeat the previous steps $m$ times (omitting the index that identifies the cluster, we obtain $m$ residuals $\tilde{R}_{i\ell}$, for $i = 1, \dots, N$, and  $\ell = 1, \dots, m$.\;\\%[0.2cm]
            
            \For{$\ell = 1$ \KwTo $m$}{
                Sort the $N$ residuals in non-decreasing order, obtaining $\tilde{R}_{(i)\ell}$\;\\[0.1cm]
                \For{$i = 1$ \KwTo $N$}{
                    Compute the percentiles 2.5\% and the 97.5\% of the ordered residuals $\tilde{R}_{(i)\ell}$ over $\ell$: $\tilde{R}_{i}^{2.5\%}$ and $\tilde{R}_{i}^{97.5\%}$, respectively.\;
                }
            }
            \vspace{0.2cm}
            \KwResult{the lower and the upper bounds for each residual $R_i$ of the original regression are given by $\tilde{R}_{i}^{2.5\%}$ and $\tilde{R}_{i}^{97.5\%}$, respectively.}%\hline
            \vspace{0.1cm}
        \end{algorithm}
    
    To reckon the impact that some subjects may have on the model fit, we now discuss the analysis of influential observations. In this paper, we focus on measures of global influence for such an intent. One route to identify influential observations is to compare the model adjustment with and without each point. The generalized Cook’s distance based on the $Q$-function, a generalization of the Cook's distance by \cite{cook1977}, measures the influence of each observation in the regression coefficients. In this sense, it performs a comparison between the MLEs, with and without a point, to catch how far apart they are. If the omission of an observation strictly affects the parameter inference, then that specific point requires further investigation. \cite{zhu2001} achieved the generalized Cook distance (GCD) measure, based on the $Q$-function for models that appreciate an EM-type algorithm. The general expression of the GCD, based on the $Q$-function, is given by
    \begin{equation*}
        GCD_{kj}(\boldsymbol{\beta}) = \left(\boldsymbol{\widehat{\beta}}_{[kj]} - \boldsymbol{\widehat{\beta}}\right)^{\top} \left\{- \ddot{Q} (\boldsymbol{\widehat{\beta}};\boldsymbol{\widehat{\beta}})\right\}  \left(\boldsymbol{\widehat{\beta}}_{[kj]} - \boldsymbol{\widehat{\beta}}\right),
    \end{equation*}
    where $\ddot{Q}(\boldsymbol{\widehat{\beta}};\boldsymbol{\widehat{\beta}}) = \dfrac{\partial^2 Q(\boldsymbol{\beta};\boldsymbol{\widehat{\beta}})}{\partial \boldsymbol{\beta} \partial \boldsymbol{\beta}^{\top}} \Bigg |_{\boldsymbol{\beta}=\boldsymbol{\widehat{\beta}}}.$ A quantity with the subscript $[kj]$ indicates a measure calculated after excluding the $j$th observation from the $k$th cluster. Thus, to sidestep an embarrassing computational burden, one should use the following one-step approximation $\boldsymbol{\widehat{\beta}}_{[kj]}^{1}$ of $\boldsymbol{\widehat{\beta}}_{[kj]}$
    $$\boldsymbol{\widehat{\beta}}_{[kj]}^{1} = \boldsymbol{\widehat{\beta}} + \left\{(\boldsymbol{X}^{\top}\boldsymbol{GX})^{-1}a_{kj}\boldsymbol{x}_{kj}\right\}\big |_{\boldsymbol{\beta}=\boldsymbol{\widehat{\beta}}},$$
    where $\boldsymbol{X}$ is the matrix containing the vectors of explanatory variables associated to the vector of parameters $\boldsymbol{\beta}$, $a_{kj} = y_{kj} - \delta_{k}\mu_{kj}$, and $\boldsymbol{G} = \mbox{diag}(\delta_{k}\mu_{kj})$, for $j = 1, \dots, n_k$, and $k = 1, \dots, q$. Hence, this approach is applied to derive the following one-step diagnostic measure of influence
    $$GCD_{kj}^{1}(\boldsymbol{\beta}) = a_{kj}^{2} \boldsymbol{x}_{kj}^{\top} (\boldsymbol{X}^{\top}\boldsymbol{GX})^{-1}\boldsymbol{x}_{kj}.$$
    We illustrate the use of the residual analysis and generalized Cook distance, based on the EM-algorithm, in the real data analysis in Section \ref{illustration}.

\section{Monte Carlo simulation}\label{simulation}

    A Monte Carlo study to assess the finite-sample performance of the EM-based estimators is conducted. For this simulation study, we have considered the logarithmic link function $g(\cdot) = \log(\cdot)$  in Expression (\ref{ms1}), which is a typical choice. However, it is significant to remark that other link functions can be employed, preferably those that guarantee the support of the model parameters. Hence, 5000 Monte Carlo replications were run, through the $\mathtt{R}$ program, with the following structure
    \begin{flalign*}
        \log\mu_{kj} = \beta_0 + \beta_1 x_{kj1} + \beta_2 x_{kj2},
        \label{ms9}
    \end{flalign*}
    for $j = 1, \dots, n_{k}$, and $k = 1, \dots, q$, with $n_{k}$ denoting the sample size of cluster $k$, where $x_{kj1}$ is normally distributed with a mean of 3.7 and a standard deviation of 0.2, while $x_{kj2}$ is generated from a Bernoulli distributed with a 0.45 success probability. The values of all regressors were kept fixed during the Monte Carlo simulation. Additionally, $\boldsymbol{\theta} = (\beta_0, \beta_1, \beta_2, \phi)^\top = (3.0, -1.25, 0.75, 0.45)^\top$ were defined (based on the modeling of a data set) considering two regressors for the response variable. We take into account three scenarios for the number of clusters. Thus, we have set $q = 2, 5, 7$ for samples with sizes $n_k = 100, 200, 300$, each. 
    
    We start the analysis of the simulation results from Table \ref{tabEM}, which comprises the empirical mean and the root mean square error (RMSE) of the parameter EM-estimates. Bearing in mind the univariate case is confirmed when each group has only one element, the increase in the sampling unit means the growth in the number of clusters. Hence, the analysis of the simulation results must follow the same path. From the results given in Table \ref{tabEM}, we can observe short bias and RMSE for all configurations ($n_k = 100$, $200$, and $300$) and sample sizes ($q = 2$, $5$, and $7$) considered. The only exception is regarding the dispersion parameter $\phi$, where its estimates had a slight bias, but which seems to decrease with the enlargement in the number of clusters.

        \begin{table}[h!]
        \centering
        \caption{Empirical mean and root mean square error (in parentheses) of the EM-estimates for $q = 2, 5, 7$ with $n_k = 100, 200, 300$ along with a normal density curve.}
        \footnotesize
        \label{tabEM}
        \setlength{\tabcolsep}{0.32cm}
        \begin{tabular}{crrr}
            \hline
            \multicolumn{1}{l}{} &  \multicolumn{1}{c}{$q=2$} &  \multicolumn{1}{c}{$q=5$} &  \multicolumn{1}{c}{$q=7$} \\ \hline
            %\rowcolor[HTML]{F5F5F5}
            \multicolumn{1}{l}{\small $n_k = 100$} & & & \\[0.2cm]%\cellcolor[HTML]{F5F5F5}
            $\beta_0$ &  2.957 &  3.001 &  2.992 \\
            &  (3.167) &  (2.080) &  (1.750) \\
            $\beta_1$ &  $-$1.245 &  $-$1.253 &  $-$1.249 \\
            &  (0.853) &  (0.562) &  (0.473) \\
            $\beta_2$ &  0.760 &  0.756 &  0.751 \\
            &  (0.379) &  (0.233) &  (0.195) \\
            $\phi$ &  0.343 &  0.382 &  0.394 \\
            &  (0.328) &  (0.252) &  (0.216) \\ \hline
            %\rowcolor[HTML]{F5F5F5}
            \multicolumn{1}{l}{\small $n_k = 200$} & & & \\[0.2cm]
            $\beta_0$ &  2.985 &  3.000 &  3.006 \\
            &  (2.408) &  (1.489) &  (1.239) \\
            $\beta_1$ &  $-$1.246 &  $-$1.252 &  $-$1.251 \\
            &  (0.649) &  (0.400) &  (0.331) \\
            $\beta_2$ &  0.754 &  0.748 &  0.751 \\
            &  (0.269) &  (0.166) &  (0.141) \\
            $\phi$ &  0.326 &  0.382 &  0.399 \\
            &  (0.315) &  (0.232) &  (0.192) \\ \hline
            %\rowcolor[HTML]{F5F5F5}
            \multicolumn{1}{l}{\small $n_k = 300$} & & & \\[0.2cm]
            $\beta_0$ &  2.985 &  2.992 &  2.985 \\
            &  (1.926) &  (1.224) &  (0.989) \\
            $\beta_1$ &  $-$1.248 &  $-$1.250 &  $-$1.249 \\
            &  (0.518) &  (0.326) &  (0.263) \\
            $\beta_2$ &  0.753 &  0.749 &  0.749 \\
            &  (0.217) &  (0.138) &  (0.112) \\
            $\phi$ &  0.310 &  0.377 &  0.395 \\
            &  (0.303) &  (0.219) &  (0.183) \\ \hline
        \end{tabular}
        \end{table}

    \begin{figure}[h!]
        \centering
        \includegraphics[width=.72\textwidth]{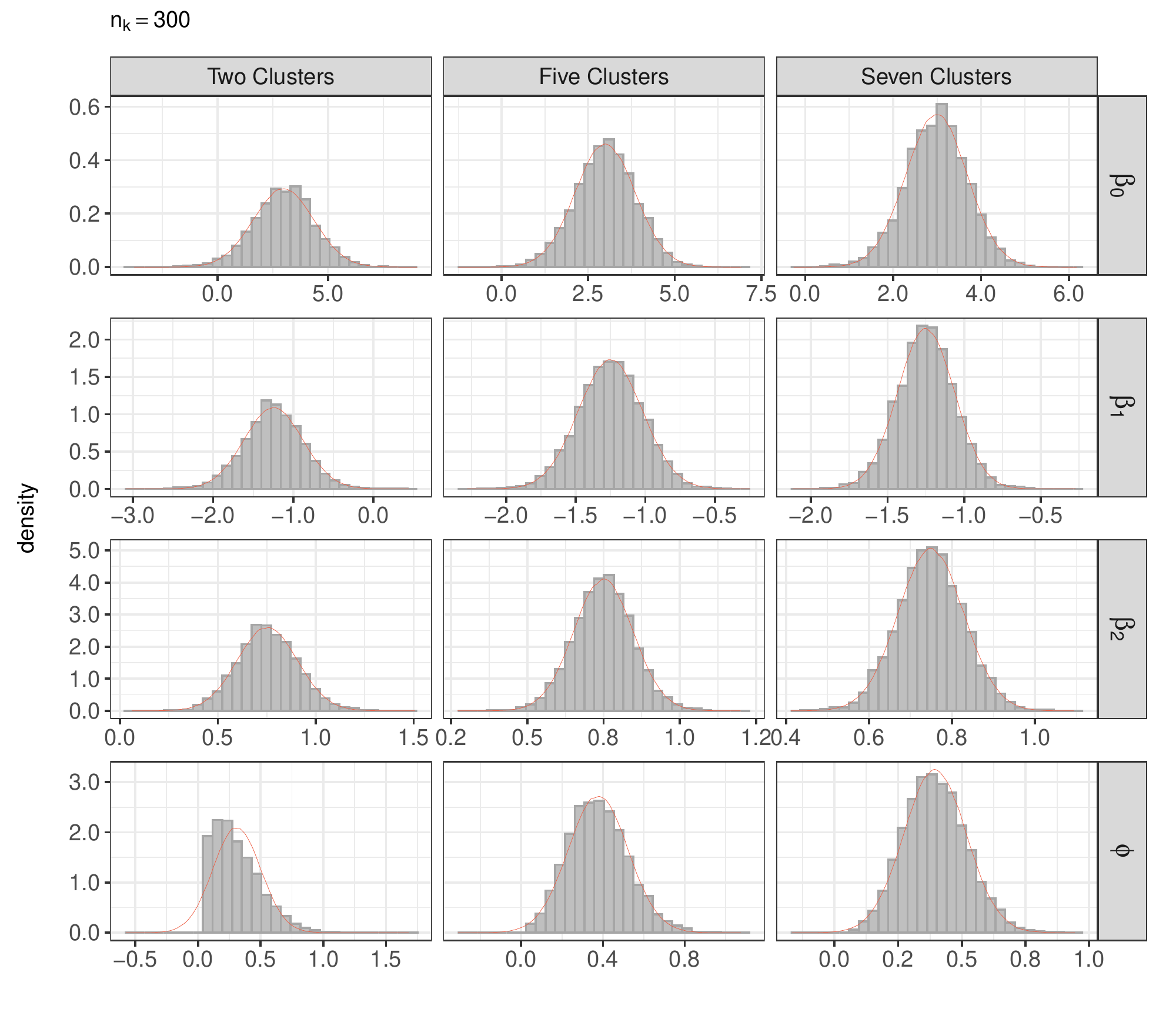}
        \caption{Histograms of the parameter EM-estimates for $q = 2, 5, 7$ with $n_k = 300$.}
        \label{fig:hist}
    \end{figure}

    Regarding the RMSE, the intercept has the highest measurements, although it decreases with clusters and, as well, with its sample sizes. The fair enactment of the simulation study results is also supported by Figure \ref{fig:hist}, which embraces the histograms of the parameter EM-estimates for $q = 2, 5, 7$ with $n_k = 300$ along with normal density curves, which reveals a good normal approximation especially when the number of clusters increases. Similar patterns were observed for the cases $n_k = 100,200$, and they are omitted to save space in the paper. 

    We conclude that the proposed EM-algorithm is working well under the configurations considered in these simulated experiments. Also, we would like to emphasize that the simulation results referring to the usual maximization of the likelihood function are similar to the results of the EM-approach. Still, some of the simulated samples failed due to numerical problems in the maximization process in almost all scenarios (the exception is the setup of seven clusters with a sample size of 300). Even though it is a small percentage of the number of Monte Carlo replications, inference via the EM-algorithm is preferable to avoid such matters.

\section{Analysis of the Medical Expenditure Panel Survey}\label{illustration}

    In this section, we motivate the CPBS regression model through the previous example enlightened at the introduction, where the goal is to model the number of inpatient admissions (response variable) from the 2003 Medical Expenditure Panel Survey (MEPS) conducted by the United States Agency for Health Research and Quality (AHRQ). The employed data set is taken from \cite{frees2009}, which is a random sample of the 2003 MEPS consisting of 2000 individuals between ages 18 and 65.
    
    The MEPS, which is considered a complete source of health care data, is a set of surveys that gathers data on the health services used by Americans, including the frequency and the costs of these services and health care coverage. For this reason, several researchers have used the MEPS for numerous purposes beyond the ideal offered by this work; for instance, see \cite{frees2009}. Note that \cite{bastos2021} also used the MEPS, applying their continuous Birnbaum-Saunders model to investigate the costs of health care services in the 2001 Medical Expenditure Panel Survey without transforming the response variable as usually done when using traditional continuous sample selection models. The literature contains many other studies that use the MEPS as a data source, and some samples of MEPS panels are available in the $\texttt{R}$ packages, such as \texttt{AER} by \cite{kleiber2008}, and \texttt{ssmrob} by \cite{zhelonkin2021}. The MEPS GitHub repository (\url{https://github.com/HHS-AHRQ/MEPS}) provides code examples for $\mathtt{R}$, $\mathtt{SAS}$, and $\mathtt{Stata}$ environments users to load and analyze whole MEPS panels. One can grasp more about the MEPS at \url{ https://www.meps.ahrq.gov/mepsweb/}.
    
    A cross-sectional data 2003 MEPS of 2000 subjects was utilized to illustrate the usefulness of our model. As previously mentioned, the response variable is the number of inpatient visits by individuals to hospital emergency rooms. Moreover, the explanatory variables consist of demographic, socioeconomic, and health-condition features of the individuals, such as \texttt{age}, \texttt{gender} (0 = \texttt{male}, 1 = \texttt{female}), \texttt{ethnicity} (0 = \texttt{other}, 1 = \texttt{black}), \texttt{marital status} (0 = \texttt{divorced or separated}, 1 = \texttt{other}), \texttt{income}, \texttt{employment status} (0 = \texttt{other}, 1 = \texttt{unemployed}), \texttt{insurance coverage} (0 = \texttt{no health insurance}, 1 = \texttt{covered by public/private health insurance}), \texttt{self-perceived physical health status} (\texttt{poor}, \texttt{good}, and \texttt{excellent} - baseline), and \texttt{any activity limitation} (0 = \texttt{no activity limitation}, 1 = \texttt{any activity limitation}). These variables are available into four clusters determined by the Midwest, Northeast, South, and West US regions, which are not balanced, having 393, 286, 764, and 557 subjects, respectively.
    
    Figure \ref{fig1} shows the number of inpatient admissions distributed by region. About 90\% of the individuals had no inpatient visit in all areas. Individuals from the Midwest and South had up to 5 and 7 inpatient visits, respectively, while Americans from the Northeast and West had up to 2 inpatient admissions. Also, the number of inpatient visits is somewhat distinct among regions. For example, the Midwest and South areas have a rate of nearly 7\% of those who had one inpatient visit, while the rates of the Northeast and West zones are close to 9\% and 5\%, respectively. In addition, the dissimilarity concerning the mean and standard deviation among the regions encourages the usage of our model.%\\%[-0.3cm]

    After a preliminary data analysis based on our CPBS regression, we selected the following covariates: \texttt{gender}, \texttt{ethnicity}, \texttt{marital status}, \texttt{employment status}, \texttt{insurance coverage}, and \texttt{self-perceived physical health status}. We begin the analysis by displaying, in Table \ref{tabfit:side}, the summary of the model's fit with the EM-based parameter estimates, the standard error estimates (based on $B = 500$ bootstrap replications), $z$-values, and associated $p$-values. 
    
    \begin{SCtable}[][hbt!]
        \caption{Parameter estimates, standard errors, $z$-values, and $p$-values for the CPBS regression model applied to the number of inpatient admissions data set.}
        \label{tabfit:side}
        \setlength{\tabcolsep}{0.4cm}
        %\centering
        \begin{tabular}{lrrrr}
        \hline
        Parameter       & \multicolumn{1}{c}{Estimate} & \multicolumn{1}{c}{S.E.} & \multicolumn{1}{c}{$z$-value} & \multicolumn{1}{c}{$p$-value} \\ \hline
        \texttt{Intercept}       & $-$4.139                 & 0.420                    & $-$9.859                    & $<$0.001                    \\
        \texttt{Female}          & 0.388                    & 0.159                    & 2.441                       & 0.015                       \\
        \texttt{Black}           & 0.347                    & 0.172                    & 2.022                       & 0.043                       \\
        \texttt{Marital Status}  & $-$0.370                 & 0.175                    & $-$2.119                    & 0.034                       \\
        \texttt{Unemployed}      & 0.712                    & 0.155                    & 4.577                       & $<$0.001                    \\
        \texttt{Insurance}       & 1.322                    & 0.301                    & 4.397                       & $<$0.001                    \\
        \texttt{Health\_Poor}    & 1.826                    & 0.270                    & 6.771                       & $<$0.001                    \\
        \texttt{Health\_Good}    & 0.369                    & 0.218                    & 1.688                       & 0.091                       \\
        $\phi$          & 0.175                    & 0.080                    &       $(--)$                 &         $(--)$                \\ \hline
        \end{tabular}
    \end{SCtable}

    The analysis of Table \ref{tabfit:side} allows us to conclude that the covariates are all significant, considering a significance level at 5\%, except the category \texttt{good health} of the variable \texttt{self-perceived physical health status}. Even so, we chose to keep this covariate instead of recategorizing, as it is an explanatory variable with three categories, and its $p$-value (equal to 0.091) is below the 10\% significance level.
    
    Continuing the modeling cycle, we are now interested in verifying if the assumed CPBS distributed response is adequate for the data set considered here. Figure \ref{envelope} presents the simulated envelopes (see Algorithm \ref{Algorithm2}) for the Pearson residual against the theoretical quantiles of the standard normal distribution. Since almost all residuals remain within the simulated envelopes, around 98.8\%, the model seems adequate for dealing with the number of inpatient admissions.
    
    \begin{figure}[hbt!]
        \centering
        \includegraphics[width=0.5\textwidth]{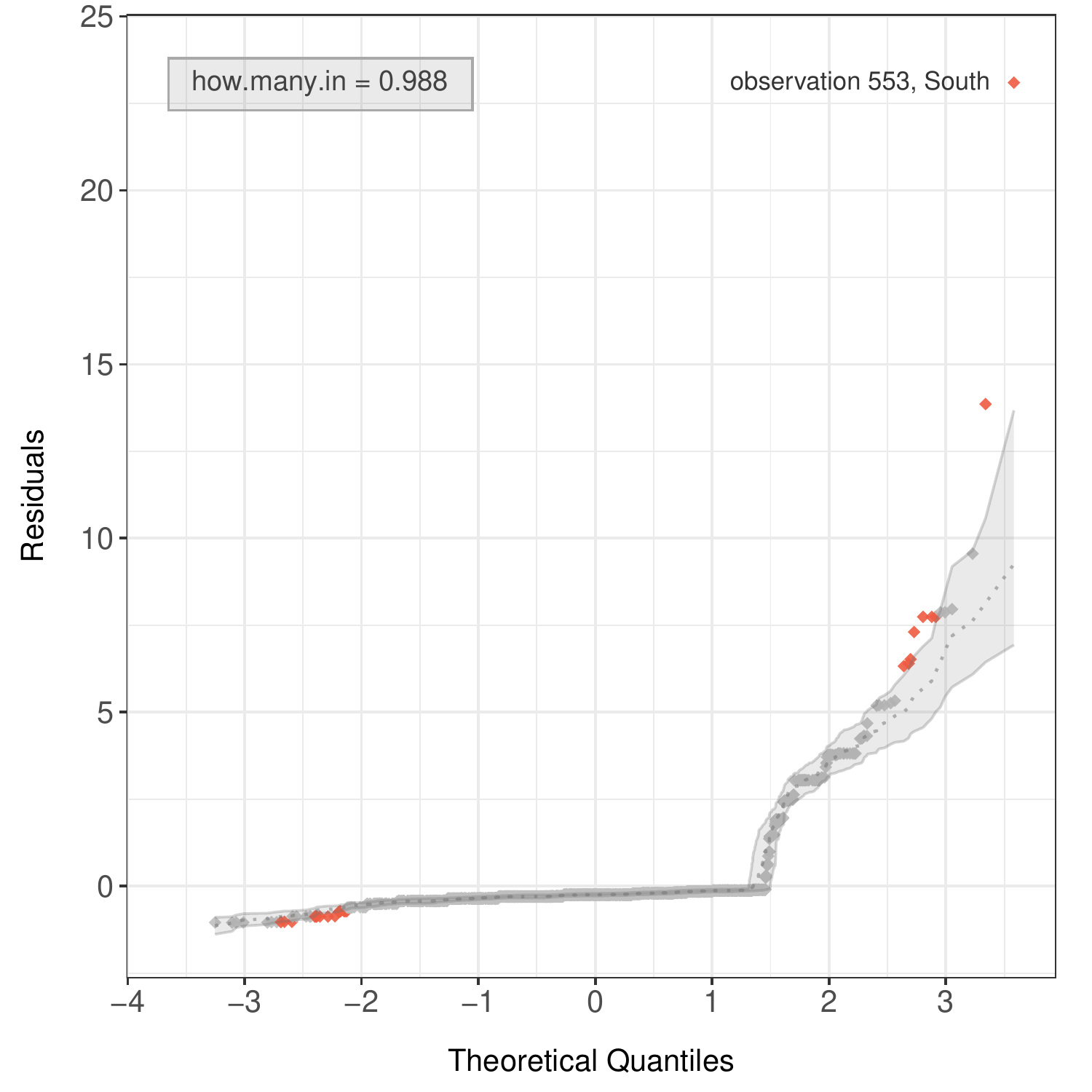}
        \caption{Simulated envelopes for the Pearson residuals under the CPBS regression for the number of inpatient admissions data set.}
        \label{envelope}
    \end{figure}
    
    Focusing on the diagnostic analysis, we now discuss the presence of influential observations through Figure \ref{GCD}, which delivers the plots of the generalized Cook's distance measure by US regions. Essentially, the plots indicate observations $\mathtt{\#51}$ and $\mathtt{\#143}$ from the Midwest region as possible influential points. In Table \ref{tabfit2:side}, we present the model's fit after excluding these observations and compare it with the fitted model using the complete data set (previously reported in Table \ref{tabfit:side}).
    
    \begin{figure}[h!]
        \centering
        \includegraphics[width=0.65\textwidth]{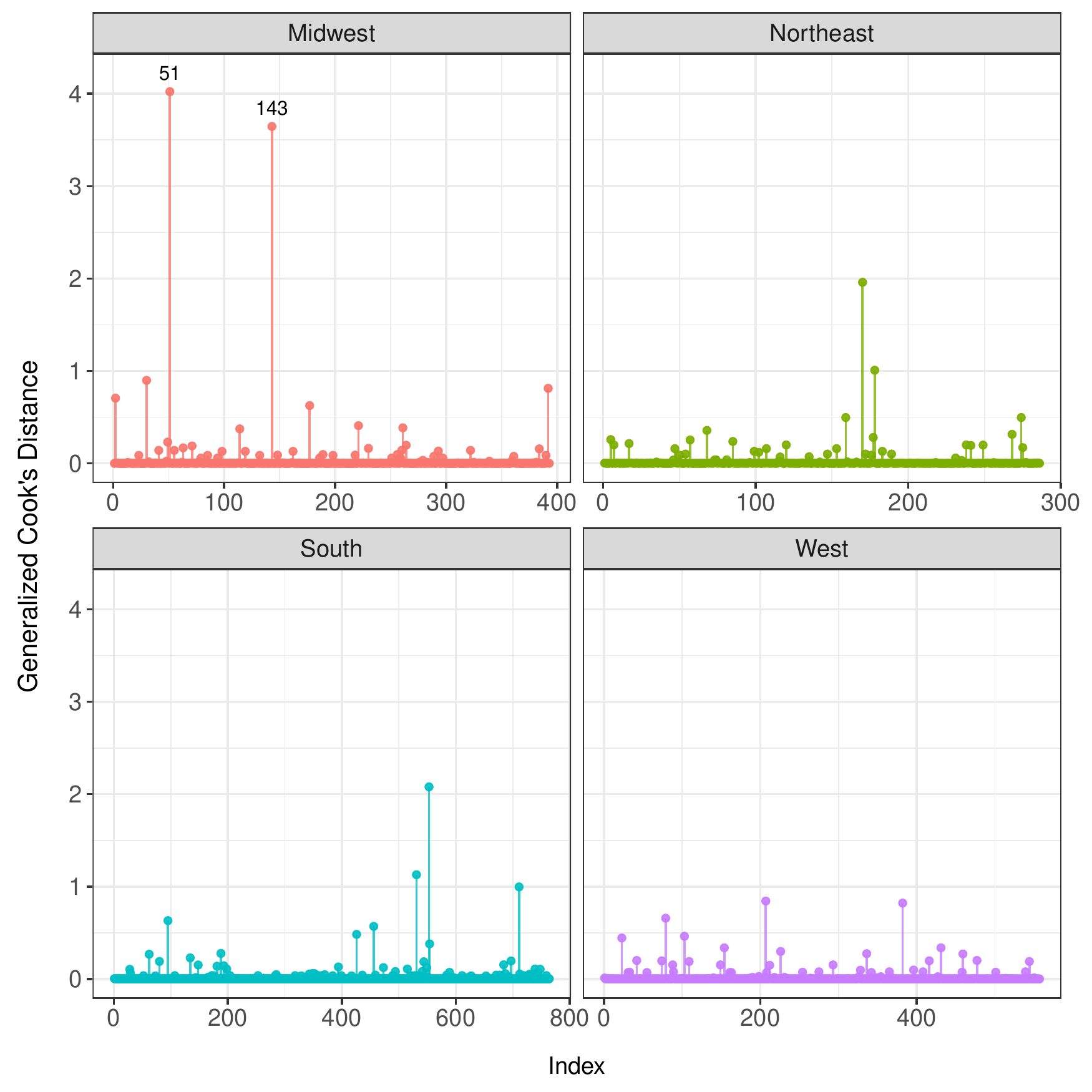}
        \caption{Generalized Cook's distance under the CPBS regression model for the number of inpatient admissions data set.}
        \label{GCD}%\\[0.5cm]
    \end{figure}
    
       \begin{SCtable}[][h!]
        \caption{Parameter estimates after excluding outliers and associated p-values (in parentheses) under the CPBS regression model.}
        \label{tabfit2:side}
        %\renewcommand{\arraystretch}{1.3}
        %\setlength{\tabcolsep}{0.225cm}
        %\centering
        \begin{tabular}{lrrr}
        \hline
        \multirow{2}{*}{Parameter} & \multicolumn{1}{c}{Estimates}   & \multicolumn{1}{c}{Estimates}     &
        \multicolumn{1}{c}{\multirow{2}{*}{Variation}} \\
        & \multicolumn{1}{c}{(full data)} & \multicolumn{1}{c}{(no outliers)} & \multicolumn{1}{c}{}                           \\ \hline
        \texttt{Intercept}      & $-$4.139   & $-$4.146   & 0.2\%     \\
                                & ($<$0.001) & ($<$0.001) &           \\
        \texttt{Female}         & 0.388      & 0.546      & 40.9\%    \\
                        & (0.015)    & (0.001)    &           \\
        \texttt{Black}          & 0.347      & 0.428      & 23.4\% \\
                        & (0.043)    & (0.010)    &           \\
        \texttt{Marital Status} & $-$0.370   & $-$0.419   & 13.1\%  \\
                        & (0.034)    & (0.011)    &           \\
        \texttt{Unemployed}     & 0.712      & 0.668      & $-$6.1\%     \\
                        & ($<$0.001) & ($<$0.001) &           \\
        \texttt{Insurance}      & 1.322      & 1.278      & $-$3.3\%  \\
                        & ($<$0.001) & ($<$0.001) &           \\
        \texttt{Health\_Poor}   & 1.826      & 1.702      & $-$6.8\%  \\
                        & ($<$0.001) & ($<$0.001) &           \\
        \texttt{Health\_Good}   & 0.369      & 0.304      & $-$17.7\% \\
                        & (0.091)    & (0.160)    &           \\
        $\phi$                  & 0.175      & 0.113      & $-$35.6\% \\ \hline\\
        \end{tabular}
    \end{SCtable}

    Analyzing the outputs of Table \ref{tabfit2:side}, we see that the estimated coefficients associated with the explanatory variables \texttt{gender}, \texttt{ethnicity}, the category \texttt{good health} of the variable \texttt{self-perceived physical health status}, and the estimate of the dispersion parameter underwent the most substantial variations after removing outliers. On the other hand, when analyzing the significance of these covariates, we observe that there is no inferential change. Therefore, these considerations lead us to conclude that the proposed model produces a robust fitting to this data set.
    
    \begin{SCtable}[][h!]
        \caption{Relativities of the explanatory variables, to estimate the mean number of inpatient visits by individuals, under the CPBS regression model.}
        \label{tab:relativ}
        \setlength{\tabcolsep}{0.9cm}
        \centering
        \begin{tabular}{lc}
            \hline
            Parameter & Relativity\\
            \hline
            \texttt{Female}                                        & 1.474                                          \\
            \texttt{Black}                                         & 1.415                                          \\
            \texttt{Marital Status}                                & 0.690                                          \\
            \texttt{Unemployed}                                    & 2.037                                          \\
            \texttt{Insurance}                                     & 3.752                                          \\
            \texttt{Health\_Poor}                                  & 6.206                                          \\
            \texttt{Health\_Good}                                  & 1.446                                          \\ \hline
        \end{tabular}
    \end{SCtable}
    
    On the interpretation of the model, we can use relativities as provided in Table \ref{tab:relativ}. These measures aim to compare a covariate's category with its baseline in terms of predicted response. Table \ref{tab:relativ} reveals that the expected number of inpatient visits is 47.4\% higher for \texttt{females} than \texttt{males}. Also, for Americans declared \texttt{black}, the expected number of inpatient visits is 41.5\% higher than other ethnicities. Making a final example of interpretation, the expected number of inpatient visits for an American who \texttt{self-perceived health} as \texttt{poor} is near six times greater than for an American who \texttt{self-perceived health} as \texttt{excellent}. The interpretation of the relativities for other covariates follows similarly.
    
    We conclude this section by checking if there are inferential changes when analyzing the data set through a model that ignores the cluster structure. We consider a univariate PBS model in this investigation, which is a particular case of our approach. Table \ref{tabfit:side2} exhibits the univariate PBS model fit summary, which ignores the variation across regions. From that table, we can observe that the covariates \texttt{ethnicity} and \texttt{marital status} are not significant (significance level at 5\%) under the univariate PBS model, in contrast with the CPBS fitting where these explanatory variables are significant. There is also a noticeable difference in the dispersion parameter estimate, which affects the probability function. Considering, for instance, all baseline categories, this implies expecting an almost 9\% reduction in the number of inpatient visits by individuals according to the univariate PBS model when compared to our clustered model (fit given in Table \ref{tabfit:side}). 
    
    \begin{SCtable}[][h!]
        \caption{Parameter estimates, standard errors, $z$-values, and $p$-values for the univariate PBS regression model applied to the number of inpatient admissions data set.}
        \label{tabfit:side2}
        \setlength{\tabcolsep}{0.4cm}
        %\centering
        \begin{tabular}{lrrrr}
        \hline
        Parameter       & \multicolumn{1}{c}{Est.} & \multicolumn{1}{c}{S.E.} & \multicolumn{1}{c}{$z$-value} & \multicolumn{1}{c}{$p$-value} \\ \hline
        \texttt{Intercept}       & $-$5.037                 & 0.536                    & $-$9.404                    & $<$0.001                    \\
        \texttt{Female}          & 0.486                    & 0.164                    & 2.962                       & 0.003                       \\
        \texttt{Black}           & 0.263                    & 0.218                    & 1.206                       & 0.228                       \\
        \texttt{Marital Status}  & $-$0.359                 & 0.205                    & $-$1.755                    & 0.079                       \\
        \texttt{Unemployed}      & 0.726                    & 0.209                    & 3.474                       & $<$0.001                    \\
        \texttt{Insurance}       & 1.342                    & 0.345                    & 3.891                       & $<$0.001                    \\
        \texttt{Health\_Poor}    & 1.931                    & 0.345                    & 5.597                       & $<$0.001                    \\
        \texttt{Health\_Good}    & 0.375                    & 0.252                    & 1.488                       & 0.137                       \\
        $\phi$          & 1.601                    & 0.349                    &        $(--)$               &      $(--)$              \\ \hline
        \end{tabular}
    \end{SCtable}

\section{Concluding remarks and future research}\label{conclusion}

    In this paper, we have proposed a new regression model to analyze clustered count data, with a Birnbaum-Saunders cluster-specific random effect, which accounts for overdispersion and dependence within the clusters. Likelihood inference based on the EM-algorithm was proposed, which overcomes possible numerical issues faced when using a direct maximization of the log-likelihood function. We also provided a measure of global influence and simulated envelopes for checking the model adequacy of our Clustered Poisson-Birnbaum-Saunders (CPBS) regression model. A random sample of the 2003 Medical Expenditure Panel Survey from the Agency for Health Research and Quality was employed to illustrate the usefulness of our regression model for analyzing clustered count data. We studied the number of inpatient admissions by individuals to hospital emergency rooms using the US regions as clusters through the proposed CPBS regression model. The clustered analysis of this count data from the MEPS is a novel contribution to the best of our knowledge. Complete data analysis was performed, showing that the CPBS model provides an adequate fit to the number of inpatient admissions by individuals. We also illustrated that ignoring the clusters can conduct inferential changes.
    
    Towards future research, noteworthy issues that deserve further investigation  are
    (i) generalization of the model allowing for a varying dispersion parameter; (ii) a zero-inflated version of the CPBS regression; (iii) to design an \texttt{R} package for fitting the CPBS model.
    
\section*{Acknowledgments} 
\noindent W. Barreto-Souza and H. Ombao acknowledge the support of the KAUST Research Fund. 
%and NIH 1R01EB028753-01. 


\begin{thebibliography}{100}
	\expandafter\ifx\csname natexlab\endcsname\relax\def\natexlab#1{#1}\fi%\setcitestyle{numbers}%\setcitestyle{numbers,square}
	
	{\small
	\bibitem[{Atkinson(1985)}]{atkinson1985} 
	\textsc{Atkinson, A.C.} (1985).
	\newblock{Plots, Transformations, and Regression}.
	\newblock{Oxford University Press: Oxford.}
	
    \bibitem[{Barreto-Souza and Simas(2016)}]{bss2016}
	\textsc{Barreto-Souza, W.} \& \textsc{Simas, A.B.} (2016).
	\newblock{General mixed Poisson regression models with varying dispersion}.
	\newblock{\emph{Statistics and Computing}.} \textbf{26}, 1263--1280.

	\bibitem[{Bastos and Barreto-Souza(2021)}]{bastos2021}
	\textsc{Bastos, F.S.} \& \textsc{Barreto-Souza, W.} (2021).
	\newblock{Birnbaum–Saunders sample selection model}.
	\newblock{\emph{Journal of Applied Statistics}.} \textbf{48}, 1896--1916.
	
	\bibitem[{Birnbaum and Saunders(1969)}]{bs1969}
	\textsc{Birnbaum, Z.W} \& \textsc{Saunders, S.C.} (1969).
	\newblock{A new family of life distributions}.
	\newblock{\emph{Journal of Applied Probability}.} \textbf{6}, 319--327.

	
	\bibitem[{Cameron and Trivedi(2013)}]{cameron2013} 
	\textsc{Cameron, A.C.} \& \textsc{Trivedi, P.K.} (2013).
	\newblock{Regression Analysis of Count Data}.
	\newblock{Cambridge University Press: Cambridge.}

	\bibitem[{Choo-Wosoba and Datta(2018)}]{choodatta2018} 
	\textsc{Choo-Wosoba, H.} \& \textsc{Datta, S.} (2018).  
	\newblock{Analyzing clustered count data with a cluster-specific random effect zero-inflated Conway–Maxwell–Poisson distribution}.
	\newblock{\emph{Journal of Applied Statistics}.} \textbf{45}, 799--814.

	\bibitem[{Choo-Wosoba et al.(2018)}]{choowosoba2018}
	\textsc{Choo-Wosoba, H.}, \textsc{Gaskins, J.}, \textsc{Levy, S.} \& \textsc{Datta, S.} (2018).
	\newblock{A Bayesian approach for analyzing zero-inflated clustered count data with dispersion}.
	\newblock{\emph{Statistics in Medicine}.} \textbf{37}, 801--812.

	\bibitem[{Choo-Wosoba et al.(2016)}]{choowosoba2016}
	\textsc{Choo-Wosoba, H.}, \textsc{Levy, S.M.} \& \textsc{Datta, S.} (2016).
	\newblock{Marginal regression models for clustered count data based on zero-inflated Conway–Maxwell–Poisson distribution with applications}.
	\newblock{\emph{Biometrics}.} \textbf{72}, 606--618.

	\bibitem[{Consul and Famoye(1992)}]{consul1992} 
	\textsc{Consul, P.C.} \& \textsc{Famoye, F.} (1992).
	\newblock{Generalized Poisson regression model}.
	\newblock{\emph{Communications in Statistics, Theory and Methods}.} \textbf{21}, 89--109.
	
	\bibitem[{Cook(1977)}]{cook1977} 
	\textsc{Cook, R.D.} (1977).  
	\newblock{Detection of influential observation in linear regression}.
	\newblock{\emph{Technometrics}.} \textbf{19}, 15--18.
	
	\bibitem[{Dean et al.(1989)}]{dean1989} 
	\textsc{Dean, C.}, \textsc{Lawless, J.F.} \& \textsc{Willmot, G.E.} (1989).  
	\newblock{A mixed Poisson-inverse-Gaussian regression model}.
	\newblock{\emph{Canadian Journal of Statistics}.} \textbf{17}, 171--181.
	
	\bibitem[{Demidenko(2007)}]{demidenko2007} 
	\textsc{Demidenko, E.} (2007).  
	\newblock{Poisson regression for clustered data}.
	\newblock{\emph{International Statistical Review}.} \textbf{75}, 96--113.	

	\bibitem[{Dempster et al.(1977)}]{dempster1977} 
	\textsc{Dempster,  A.P.}, \textsc{Laird, N.M.} \& \textsc{Rubin, D.B.} (1977).  
	\newblock{Maximum likelihood from incomplete data via the EM algorithm}.
	\newblock{\emph{Journal of the Royal Statistical Society - Series B}.} \textbf{39}, 1--38.

	\bibitem[{Dimitris and Xekalaki(2005)}]{sellers2005} 
	\textsc{Dimitris, K.} \& \textsc{Xekalaki, E.} (2005).
	\newblock{Mixed Poisson distributions}.
	\newblock{\emph{International Statistical Review}.} \textbf{73}, 35--58.	
	
	\bibitem[{Efron(1979)}]{efron1979} 
	\textsc{Efron, B.} (1979).  
	\newblock{Bootstrap methods: Another look at the Jackknife}.
	\newblock{\emph{Annals of Statistics}.} \textbf{7}, 1--26.	
	
	\bibitem[{Efron and Tibshirani(1994)}]{efron1994} 
	\textsc{Efron, B.} \& \textsc{Tibshirani, R.J.} (1994).
	\newblock{An Introduction to the Bootstrap}.
	\newblock{Chapman and Hall.}

	\bibitem[{Fabio et al.(2012)}]{fabio2012} 
	\textsc{Fabio, L.C.}, \textsc{Paula, G.A.} \& \textsc{Castro, M.} (2012).  
	\newblock{A Poisson mixed model with nonnormal random effect distribution}.
	\newblock{\emph{Computational Statistics \& Data Analysis}.} \textbf{56}, 1499--1510.

	\bibitem[{Famoye and Singh(2006)}]{famoye2006} 
	\textsc{Famoye, F.} \& \textsc{Singh, K.P.} (2006).
	\newblock{Zero-inflated generalized Poisson regression model with an application to domestic violence data}.
	\newblock{\emph{Journal of Data Science}.} \textbf{4}, 117--130.

	\bibitem[{Frees(2009)}]{frees2009} 
	\textsc{Frees, E.W.} (2009).
	\newblock{Regression Modeling with Actuarial and Financial Applications (International Series on Actuarial Science)}. \newblock{Cambridge University Press: Cambridge.}
	
	\bibitem[{Gómez-Déniz et al.(2016)}]{gupta2016} 
	\textsc{Gómez-Déniz, E.}, \textsc{Ghitany, M. E.} \& \textsc{Gupta, R. C.} (2016).  
	\newblock{Poisson-mixed inverse Gaussian regression model and its application}.
	\newblock{\emph{Communications in Statistics - Simulation and Computation}.} \textbf{45}, 2767--2781.
	
	\bibitem[{Gonçalves and Barreto-Souza(2020)}]{goncalves2020} 
	\textsc{Gonçalves, J.N.} \& \textsc{Barreto-Souza, W.} (2020).
	\newblock{Flexible regression models for counts with high-inflation of zeros}.
	\newblock{\emph{Metron}.} \textbf{78}, 71--95.

	\bibitem[{Guo(1996)}]{guo1996} 
	\textsc{Guo, G.} (1996).  
	\newblock{Negative multinomial regression models for clustered event counts}.
	\newblock{\emph{Sociological Methodology}.} \textbf{26}, 113--132.
	
	\bibitem[{Hall(2000)}]{hall2000} 
	\textsc{Hall, D.B.} (2000).  
	\newblock{Zero-inflated Poisson and binomial regression with random effects: A case study}.
	\newblock{\emph{Biometrics}.} \textbf{56}, 1030--1039.

	\bibitem[{Hall and Zhang(2004)}]{hall2004} 
	\textsc{Hall, D.B.} \& \textsc{Zhang, Z.} (2004).  
	\newblock{Marginal models for zero inflated clustered data}.
	\newblock{\emph{Statistical Modelling}.} \textbf{4}, 161--180.	

	\bibitem[{Hilbe(2007)}]{hilbe2007} 
	\textsc{Hilbe, J.M.} (2007).
	\newblock{Negative Binomial Regression}.
	\newblock{Cambridge University Press: Cambridge.}
	
	\bibitem[{Hinde and Dem{\'e}trio(1998)}]{hinde1998}
	\textsc{Hinde, J.} \& \textsc{Dem{\'e}trio, C.G.B.} (1998).
	\newblock{Overdispersion: models and estimation}.
	\newblock{\emph{Computational Statistics \& Data Analysis}.} \textbf{27}, 151--170.
	
	\bibitem[{Holla(1966)}]{holla1966}
	\textsc{Holla, M.S.} (1966).
	\newblock{On a Poisson-inverse Gaussian distribution}.
	\newblock{\emph{Metrika}.} \textbf{11}, 115--121.

	\bibitem[{Kang et al.(2021)}]{kang2021}
	\textsc{Kang, T.}, \textsc{Levy, S.M.} \& \textsc{Datta, S.} (2021).
	\newblock{Analyzing longitudinal clustered count data with zero inflation: Marginal modeling using the Conway–Maxwell–Poisson distribution}.
	\newblock{\emph{Biometrical Journal}.} \textbf{63}, 761--786.
	
	\bibitem[{Kleiber and Zeileis(2008)}]{kleiber2008} 
	\textsc{Kleiber, C.} \& \textsc{Zeileis, A.} (2008).
	\newblock{Applied Econometrics with {R}}.
	\newblock{Springer-Verlag: New York.}
	\newblock{Available at \url{https://CRAN.R-project.org/package=AER}}.
	
	\bibitem[{Lambert(1992)}]{lambert1992} 
	\textsc{Lambert, D.} (1992).  
	\newblock{Zero-inflated Poisson regression, with an application to defects in manufacturing}.
	\newblock{\emph{Technometrics}.} \textbf{34}, 1--14.	

	\bibitem[{Lawless(1987)}]{lawless1987} 
	\textsc{Lawless, J.F.} (1987).  
	\newblock{Negative binomial and mixed Poisson regression}.
	\newblock{\emph{Canadian Journal of Statistics}.} \textbf{15}, 209--225.
	
	\bibitem[{Louis(1982)}]{louis1982} 
	\textsc{Louis, T.A.} (1982).  
	\newblock{Finding the observed information matrix when using the EM algorithm}.
	\newblock{\emph{Journal of the Royal Statistical Society - Series B}.} \textbf{44}, 226--233.

	\bibitem[{Ma et al.(2009)}]{ma2009} 
	\textsc{Ma, R.}, \textsc{Hasan, M.T.} \& \textsc{Sneddon, G.} (2009).  
	\newblock{Modelling heterogeneity in clustered count data with extra zeros using compound Poisson random effect}. \newblock{\emph{Statistics in Medicine}.} \textbf{28}, 2356--2369.

	\bibitem[{R Core Team(2021)}]{R2021} 
	\textsc{{R Core Team}} (2021).
	\newblock{{R}: A Language and Environment for Statistical Computing}.
	\newblock{R Foundation for Statistical Computing: Vienna, Austria.}
	\newblock{Available at \url{https://www.R-project.org/}}.

	\bibitem[{Ridout et al.(1998)}]{ridout1998}
	\textsc{Ridout, M.}, \textsc{Hinde, J.} \& \textsc{Dem{\'e}trio, C.G.B.} (1998).
	\newblock{Models for count data with many zeros. In: Proceedings of the XIXth International Biometrics Conference}. \newblock{Cape Town, Invited Papers.}  179--192.%\textbf{57},

	\bibitem[{Ridout et al.(2001)}]{ridout2001}
	\textsc{Ridout, M.}, \textsc{Hinde, J.} \& \textsc{Dem{\'e}trio, C.G.B.} (2001).
	\newblock{A score test for testing a zero-inflated Poisson regression model against zero-inflated negative binomial alternatives}.
	\newblock{\emph{Biometrics}.} \textbf{57}, 219--223.

	\bibitem[{Rosen et al.(2000)}]{rosen2000} 
	\textsc{Rosen, O.}, \textsc{Jiang, W.} \& \textsc{Tanner, M.A.} (2000).  
	\newblock{Mixtures of marginal models}.
	\newblock{\emph{Biometrika}.} \textbf{87}, 391--404.

	\bibitem[{Sellers and Shmueli(2010)}]{sellers2010} 
	\textsc{Sellers, K.F.} \& \textsc{Shmueli, G.} (2010).
	\newblock{A flexible regression model for count data}.
	\newblock{\emph{Annals of Applied Statistics}.} \textbf{4}, 943--961.
	
	\bibitem[{Shoukri et al.(2004)}]{shoukri2004}
	\textsc{Shoukri, M.M.}, \textsc{Asyali, M.H.}, \textsc{VanDorp, R.} \& \textsc{Kelton, D.} (2004).
	\newblock{The Poisson inverse Gaussian regression model in the analysis of clustered counts data}.
	\newblock{\emph{Journal of Data Science}.} \textbf{2}, 17--32.

	\bibitem[{Willmot(1987)}]{willmot1987} 
	\textsc{Willmot, G.E.} (1987).  
	\newblock{The Poisson-inverse Gaussian distribution as an alternative to the negative binomial}.
	\newblock{\emph{Scandinavian Actuarial Journal}.} \textbf{1987}, 113--127.
	
%	\bibitem[{Wu(1983)}]{wu1983} 
%	\textsc{Wu, C.F.J.} (1983).  
%	\newblock{On the convergence properties of the EM algorithm}.
%	\newblock{\emph{Annals of Statistics}.} \textbf{11}, 95--103.

	\bibitem[{Yau et al.(2003)}]{yau2003}
	\textsc{Yau, K.K.W.}, \textsc{Wang, K.} \& \textsc{Lee, A.H.} (2003).
	\newblock{Zero-inflated negative binomial mixed regression modeling of over-dispersed count data with extra zeros}. \newblock{\emph{Biometrical Journal}.} \textbf{45}, 437--452.
	
	\bibitem[{Zhelonkin and Ronchetti(2021)}]{zhelonkin2021} 
	\textsc{Zhelonkin, M.} \& \textsc{Ronchetti, E.} (2021).  
	\newblock{Robust analysis of sample selection models through the {R} package {ssmrob}}.
	\newblock{\emph{Journal of Statistical Software}.} \textbf{99}, 1--35.
	\newblock{Available at \url{https://CRAN.R-project.org/package=ssmrob}}.
	
	\bibitem[{Zhu et al.(2001)}]{zhu2001} 
	\textsc{Zhu, H.}, \textsc{Lee, S.Y.}, \textsc{Wei, B.C.} \& \textsc{Zhou, J.} (2001).
	\newblock{Case-deletion measures for models with incomplete data}.
	\newblock{\emph{Biometrika}.} \textbf{88}, 727--737.
	}    		
\end{thebibliography}
\end{document}